\documentclass[11pt,lettersize]{article}

\usepackage{geometry}
\usepackage{algcompatible}
\usepackage{amssymb}

\usepackage{tikz}

\algnewcommand\algorithmicreturn{\textbf{return}}
\algnewcommand\RETURN{\algorithmicreturn}
\algnewcommand\algorithmicprocedure{\textbf{procedure}}
\algnewcommand\PROCEDURE{\item[\algorithmicprocedure]}%
\algnewcommand\algorithmicendprocedure{\textbf{end procedure}}
\algnewcommand\ENDPROCEDURE{\item[\algorithmicendprocedure]}%
\algnewcommand{\algvar}[1]{{\text{\ttfamily\detokenize{#1}}}}
\algnewcommand{\algarg}[1]{{\text{\ttfamily\itshape\detokenize{#1}}}}
\algnewcommand{\algproc}[1]{{\text{\ttfamily\detokenize{#1}}}}
\algnewcommand{\algassign}{\leftarrow}

\newtheorem{theorem}{Theorem}[section]
\newtheorem{definition}{Definition}[section]
\newtheorem{corollary}{Corollary}[section]

\newtheorem{lemma}{Lemma}[section]
\newtheorem{observation}{Observation}[section]

\newcommand{\qed}{\hfill\hbox{\rlap{$\sqcap$}$\sqcup$}}
\newenvironment{proof}{\noindent \emph{Proof.\,}}{\qed}

\usepackage{graphicx} 
\usepackage{amsmath,psfrag} 
\usepackage{algorithm,cleveref}
\usepackage{algpseudocode}
\usepackage{algpseudocodex}

\usepackage[leftcaption]{sidecap}

\usepackage{mathtools}

\usepackage{complexity}
\usepackage[misc]{ifsym}
\usepackage{commands}
\usepackage{cancel}
\usepackage{xspace}

\usepackage{commands}
\usepackage{todonotes}

\usepackage{standalone}
\usepackage{tikz}

\title{Computing and Enumerating Minimal Common Supersequences Between Two Strings} 

\author{
Braeden Sopp\footnote{Gianforte School of Computing, Montana State University, Bozeman, MT 59717, USA. Email: {\tt braeden.sopp@student.montana.edu}}
\and
Adiesha Liyanage\footnote{Gianforte School of Computing, Montana State University, Bozeman, MT 59717, USA. Email: {\tt a.liyanaralalage@montana.edu}}
\and
Mingyang Gong\footnote{Gianforte School of Computing, Montana State University, Bozeman, MT 59717, USA. Email: {\tt mingyang.gong@montana.edu}}
\and
Binhai Zhu\footnote{Gianforte School of Computing, Montana State University, Bozeman, MT 59717, USA.
Email: {\tt bhz@montana.edu}}
}

\date{}

\begin{document}

\maketitle

\begin{abstract}
Given \(k\) strings each of length at most $n$, 
computing the shortest common supersequence of them is a well-known NP-hard problem (when \(k\) is unbounded). On the other hand, when \(k=2\), such a shortest common supersequence can be computed in \(O(n^2)\) time using dynamic programming as a textbook example. In this paper, we consider the problem of computing a \emph{minimal} common supersequence and enumerating all minimal common supersequences for \(k=2\) input strings. Our results are summarized as follows.
\begin{enumerate}
    \item A minimal common supersequence of \(k=2\) input strings can be computed in $O(n)$ time. (The method also
    works when \(k\) is a constant).
    \item All minimal common supersequences between two input strings can be enumerated with a data structure of $O(n^2)$ space and an $O(n)$ time delay, and the data structure can be constructed in $O(n^3)$ time.
\end{enumerate}
\end{abstract}

\section{Introduction}

Computing the longest common subsequences (LCS) in two or more input strings is a classical problem which has found many applications. The textbook example is the longest common subsequence of two
strings (LCS-2) each of length at most $n$, which was solved by Wagner and Fischer in $O(n^2)$ time and space \cite{DBLP:journals/jacm/WagnerF74} and then in $O(n^2)$ time but $O(n)$ space by Hirschberg \cite{DBLP:journals/cacm/Hirschberg75}. (In 2015, the $O(n^2)$ time bound was shown to be conditionally optimal under the Strong Exponential Time Hypothesis \cite{DBLP:conf/focs/AbboudBW15}, SETH for short.)  It has been investigated and applied in various applications, sometimes with additional
constraint, e.g., LCS-2 must also contain a (non-contiguous) pattern $P$, which is a subsequence appearing in both of the input strings $X$ and $Y$. If $X$, $Y$ and $P$ are all of length $O(n)$, Tsai gave an $O(n^5)$ time algorithm \cite{DBLP:journals/ipl/Tsai03}. The bound was improved to $O(n^3)$ by Chin et al. \cite{DBLP:journals/ipl/ChinSFHK04} and also by Chen and Chao \cite{DBLP:journals/jco/ChenC11}, and the latter additionally considered the version where LCS-2 of $X$ and $Y$ must not include $P$ as a subsequence. 

The shortest common supersequence (SCS) problem is also a well-known problem.
SCS finds applications in job scheduling, data merging (DNA sequence merging) and in multiple sequence alignment. (In fact the first NP-hardness proof on multiple sequence alignment is exactly reduced from SCS \cite{DBLP:journals/jcb/WangJ94}.) SCS was first studied by Maier in 1978, as the complement to the LCS problem \cite{DBLP:journals/jacm/Maier78}. Of course, compared with the large number of references on the LCS research, those on SCS are quite limited. First of all, it is a also textbook example that SCS-$d$,
where $d$ is the number of input sequences, each of length at most $n$, is also polynomially solvable when $d$ is fixed (e.g., with dynamic programming). Maier proved that both SCS and LCS are NP-complete when the number of sequences is part of the input (i.e., is not fixed) \cite{DBLP:journals/jacm/Maier78}. (For LCS, his reduction is from Vertex Cover, which can be replaced by its dual Independent Set --- IS for short, hence the inapproximability is equivalent to that of Maximum Independent Set, MIS for short \cite{Zhu07}.) 

In 1995, Jiang and Li studied the approximability of LCS and SCS \cite{JL95}. For LCS, their result was very similar to that of Maier: they showed that LCS is as hard to approximate as Maximum Clique (which is complement to MIS). For SCS, they proved two inapproximability results:
(1) SCS does not have a PTAS unless P=NP, using the NP-hardness
proof by Timkovskii \cite{VT89}, (2) if $n$ is the number of input sequences for SCS, then there is no factor-$\log^{\delta}n$ approximation for some positive constant $\delta$, unless NP is in {\bf DTIME}$(2^{\text{polylog}~n})$. 
For SCS, Pietrzak was the first to study if it admits any FPT (fixed-parameter tractable) algorithm \cite{DBLP:journals/jcss/Pietrzak03}. He proved that when the
alphabet is of a constant size and the parameter is the number of input sequences, then the SCS problem is W[1]-hard.
Dondi and, more recently, Chen et al. gave additional results on the FPT tractability of the SCS problem \cite{DBLP:journals/jda/Dondi13,DBLP:conf/aaim/ChenJLWZ24}.

While the longest common subsequence problem is well-studied, in 2018 Sakai first proposed the {\em maximal} common subsequence problem between two strings $A$ and $B$ each of length at most $n$ \cite{Sakai18}. The motivation is that in many practical applications the longest common subsequence might not capture the optimal solution due to real constraints. Consequently, the
problem was solved in $O(n\log n)$ time \cite{Sakai18,Sakai19}. \footnote{As pointed by Hirota and Sakai in \cite{DBLP:journals/tcs/HirotaS25}, the claimed bounds there were lower, due to incorrectly counting the cost for data construction.} Then, Hirota and Sakai showed that with $k$ input strings, the maximal common subsequence can be computed in $O(kN\log N)$ time, where $N$ 
is the sum of the lengths of the $k$ strings \cite{HirotaSakai23}.

Almost around the same time, the enumeration of maximal common subsequences of two input strings have been studied by Conte et al. \cite{DBLP:journals/algorithmica/ConteGPU22}. In this case, the cost is a triple $(p,s,d)$, where $O(p)$ is the time to construct a data structure of size $O(s)$ such that each maximal common subsequence is enumerated with a delay time of $O(d)$. The result by Conte et al. is $(\sigma n^2\log n, n^2, \sigma n\log n)$ ($\sigma$ is the size of the alphabet) \cite{DBLP:journals/algorithmica/ConteGPU22}. In 2025, Hirota and Sakai improved this bound with three results: $(n^2,n^2,n)$, $(n^2,n,n\log n)$ and $(n^3,n^3,n)$; the last one, though not really an improvement, gives a conceptually simple characterization of maximal common subsequences through a DAG $G$ with a source $s$ and a sink $t$, where each $st$-path in $G$ corresponds to a unique maximal common subsequence \cite{DBLP:journals/tcs/HirotaS25}. 

In this paper, we study the minimal common supersequence (MCS) problem of two input strings $A$ and $B$, each of length at most $n$, which involves both computation and enumeration.
A naive idea for the computation problem is that given a maximal common subsequence of $A$ and $B$, one could obtain dually an MCS of $A$ and $B$ --- in $O(n\log n)$ time using Sakai's result.
However, we will present a linear time algorithm to compute an MCS of two input strings.

For the enumeration problem, it should be noted that there is no correspondence between the
maximal common subsequences and MCS's of $A$ and $B$. The following is an example: $A=xay$ and $B=zaw$, $a$ is the only
common subsequence of $A$ and $B$, but $A\circ B=xay\circ zaw$ is certainly an MCS of $A$ and $B$
which is not corresponding to $a$. Therefore, we need to exploit additional properties
to enumerate all MCS's of $A$ and $B$.

Given $k$ input strings each of length at most $n$, our results are summarized as follows:
\begin{enumerate}
    \item An MCS of \(k=2\) input strings can be computed in $O(n)$ time. 
    \item An MCS of $k\geq 3$ input strings can be computed in $O(kn(\log k+\log n))$ time.
    \item All MCS's between two input strings can be enumerated with a data structure of $O(n^2)$ space and an $O(n)$ time delay, and the data structure can be constructed in $O(n^3)$ time.
\end{enumerate}

The paper is organized as follows. In Section 2 we give necessary definitions. In Section 3 we explore some properties of MCS. In Section 4 we give a linear time algorithm for computing an MCS of two strings ($k$ strings
in Section 6). In Section 5, we present the algorithm and data structure to enumerate all MCS's between two strings. We conclude the paper in Section 7. 

\section{Preliminaries}

Let \(\Sigma\) be an alphabet and \(S\) be a sequence over \(\Sigma\). We denote the length of \(S\) by \(n = |S|\) and the size of the alphabet by \(m = |\Sigma|\). We write \([n] = \{1,2,\ldots,n\}\), and for each \(i \in [n]\), let \(S[i]\) denote the \(i\)-th character of \(S\). A sequence \(X\) is a \emph{subsequence} of \(S\) if there exist indices \(1 \leq i_1 < i_2 < \cdots < i_{|X|} \leq n\) such that \(X = S[i_1] S[i_2] \cdots S[i_{|X|}]\), in which case we write \(X \subseteq S\). Similarly, a sequence \(Y\) of length \(\ell \geq n\) over \(\Sigma\) is a \emph{supersequence} of \(S\) if there exist indices \(1 \leq j_1 < j_2 < \cdots < j_n \leq \ell\) such that \(S = Y[j_1] Y[j_2] \cdots Y[j_n]\), and we write \(Y \supseteq S\). We denote the empty string as \(\epsilon\). 

Given two positive integers \(1 \leq i < j \leq |S|\), we denote the substring of \(S\) that begins with \(S[i]\) and ends with \(S[j]\) by \(\substring{S}{i}{j}\). Note that we define \(\substring{S}{j}{i}= \epsilon\) when \(j > i\). Furthermore, \(\substring{S}{1}{j}\) is called a {\em prefix} of \(S\), while \(\substring{S}{i}{n}\) is called a {\em suffix} of \(S\).  We use the open interval notation to indicate the exclusive index substrings; \ie\(S\rointv{i}{j} = \substring{S}{i}{j-1}\), \(S\lointv{i}{j} = \substring{S}{i+1}{j}\), and \(S\ointv{i}{j} = \substring{S}{i+1}{j-1}\).

Given two sequences \(S\) and \(T\) over the alphabet \(\Sigma\), a sequence $C$ is a \emph{common subsequence} of $S$ and $T$ if $C \subseteq S$ and $C \subseteq T$. It is \emph{maximal} if no proper subsequence of $C$ is still a common subsequence of $S$ and $T$. The \emph{longest common subsequence (LCS)} is a common subsequence of maximum possible length. A sequence $C$ is a \emph{common supersequence} of $S$ and $T$ if $S \subseteq C$ and $T \subseteq C$. It is \emph{minimal} if no proper subsequence of $C$ is still a common supersequence of $S$ and $T$. The \emph{shortest common supersequence (SCS)} is a common supersequence of minimum possible length.


We define \(\mcsupenum\) as the problem of enumerating all minimal common supersequences (MCS's) of two strings \(A\) and \(B\).
\begin{definition}[\(\mcsupenum\)]
    Given two strings \(A\) and \(B\) over an alphabet \(\Sigma\), the \(\mcsupenum\) problem asks to enumerate the set \(\mcsup(A, B)\) of supersequences of \(A\) and \(B\) that are minimal with respect to the subsequence relation.
    \label{def:mcsup}
\end{definition}

We make the following observation about any common supersequence \(S\) of the strings \(A\) and \(B\).

\begin{observation}
Any common supersequence $S$ of two strings $A$ and $B$ is either a minimal common supersequence or $S$ can be reduced to a minimal common supersequence $S'$ of $A$ and $B$ by deleting some letters.    
\end{observation}

Following the above observation, we can state the following lemma about MCS \(S\) of two strings \(A\)  and \(B\).

\begin{lemma}
    Given strings \(A, B\), let \(i \in \intv{1}{|S|}\), we define \(S' = S\rointv{1}{i}\circ S\lointv{i}{|S|}\). We say that \(S\) is a minimal common supersequence of \(A\) and \(B\) if and only if \(\forall i \in \intv{1}{|S|}: S'(i) \text{~is not a common supersequence of } A \text{~and~}B.\)
    \label{lemma:everyinterspaceisclean}
\end{lemma}

Moreover, we can generalize \cref{lemma:everyinterspaceisclean} for MCS's 
of \(k\) strings as well. We define an embedding of string \(A\) into \(B\) as follows.

\begin{definition}[Embedding of a string]
    Given strings \(A=a_1a_2\ldots a_n\) and \(B=b_1b_2\ldots b_m\) over an alphabet \(\Sigma\) such that \(B\) is a supersequence of \(A\), an embedding \(\Phi\) of \(A\) into \(B\) is a function \(\Phi: [n] \rightarrow [m]\) such that \(A[i] = B[\Phi[i]]\) and \(i < j \iff \Phi(i) < \Phi(j)\). Note that this is an injective function.     Whenever necessary, we extend the domain of an embedding to include \(0\) and \(|A|+1\) where \(\Phi(0) =0\) and \(\Phi(|A|+1)=|B|+1\).
\end{definition}
We consider two special embeddings, called the left and the right embeddings, defined as follows.

\begin{definition}[Left (Right) Embedding]
     Given strings \(A=a_1a_2\ldots a_n\) and \(B=b_1b_2\ldots b_m\) over an alphabet \(\Sigma\) such that \(B\) is a supersequence of \(A\) and an embedding \(\psi\) of \(A\) into \(B\), we say that \(\psi\) is the left (right) embedding if for all \(i \in [n],\) \(\psi(i)\) is defined to be the smallest (largest) index of \(B\) such that \(A[1\ldots i] \subseteq S[1\ldots \psi(i)]\) \((A[i\ldots n] \subseteq S[\psi(i)\ldots m])\). Further, we denote left (right) embedding of \(A\) into \(B\) using \(\lembd{A}{B}\) \((\rembd{A}{B})\).

\end{definition}

We use interval of real numbers to describe ranges of the string. 
Recall there are four types of intervals over real numbers with the forms
    \( (a,b), [a,b)\), \((a,b]\), or \([a,b]\)
    where \(a,b \in \mathbb{R}.\)
The value \(a\) is referred to as the \emph{left} endpoint 
and \(b\) as the \emph{right} endpoint.
We denote the set of all intervals by \(\setint\).
It is clear that \(\setint\) has a natural partial order inherited from 
the subset relation thus when we say an interval is maximal, we are referring to this ordering.

We say interval \(I\) \emph{contains no indices} if \(I \cap \Z = \emptyset\)
and we use the indices in \(I\) to refer to \(I \cap \mathbb{Z}\).
Given a string \(A \in \Sigma^{*}\), we use the indices of \(A\) to refer to the set \([|A|]\).
We write \(\interval{A}\) to denote the set
of intervals with endpoints only in \(\intv{0}{|A|+1}\) and \(\itv{A}\)
to denote the interval \([0,|A|+1]\).
Given \(I \in \interval{A}\), we write \(A[I]\) 
to denote the substring of \(A\) using indices of \(A\) in \(I\).

\section{Properties of Minimal Common Supersequences}

Before proceeding further, we give a strong characterization of minimal common supersequences in terms of \emph{essential} indices. An index in a supersequence is said to be essential for a string \(A\) if deleting that index from the supersequence yields a string that no longer contains \(A\) as a subsequence.

Further, an essential index for \(A\) in \(S\)  corresponds to a 
particular index of \(A\).

\begin{definition}\label{def:essential}
    Let \(A \subseteq S\) and \(i \in \intv{1}{|S|}\).
    We say \(i\) is essential for \(A\) in \(S\) if
    \(A \nsubseteq S' = \delchar{S}{i}\). 
    Further for \(j \in \intv{1}{|A|}\),
    we say \(i\) is essential for pair \((j, A)\) in \(S\) if
    for \(A' = \delchar{A}{j}\), 
    we have \(A' \subseteq S'\) and \(A \nsubseteq S'\).
    Note, \(i\) is essential for \(A\) in \(S\) if and only if
     \(i\) is essential for some pair \((j, A)\) in \(S\). 
\end{definition}

An example is as follows: $S=abcbacb$ and $A=abab$, index $i=5$ is essential for $A$, as $A\nsubseteq S'=abcbcb$. The index $i=5$ is also essential for $(3,A)$ as $A'=abb \subseteq S'$ but $A\nsubseteq S'$.


\begin{lemma}\label{lem:ess-iff}
    Let \(A_1, \ldots, A_k\) be strings in \(\Sigma^*\) and let \(S\) be one of their common supersequences.
 The string \(S\) is a minimal common supersequence if and only if 
 every index \(i \in \intv{1}{|S|}\)
 is essential for some \(A_j\).
\end{lemma}
\begin{proof}
    We prove the forward direction using the contrapositive.
    Let \(i \in \intv{1}{|S|}\) that is not essential for any \(A_j\) and
    consider \(S' = \delchar{S}{i}\). Note as \(A_j \subseteq S\) for any 
    \(j\), we have \(S'\) is one of their common supersequences and
    thus \(S\) is not minimal.

    For the reverse direction, note that if every index is essential then
    deleting any index forms \(S\) to produces a supersequence which does
    not contain some \(A_i\), which is sufficient to show \(S\) is minimal.
\end{proof}

The following lemma provides an easy condition to check
if an index of $S$ is essential for a pair on $A$.

\begin{lemma}\label{lem:ess-check}
    Let \(A \subseteq S\) and \(i \in \intv{1}{|S|}\) 
    and \(j \in \intv{1}{|A|}\).
    The index \(i\) is essential for pair \((j,A)\) in \(S\) if and only if
    \(\lembd{A}{S}(j) = i = \rembd{A}{S}(j)\).
\end{lemma}
\begin{proof}
    Throughout this proof let \(S' = \delchar{S}{i}\) and
    \(A' = \delchar{A}{j}\).
    
    We start with the reverse direction. 
    Given \(\lembd{A}{S}(j) = i = \rembd{A}{S}(j)\), we have
    \(A\rointv{1}{j} \subseteq S\rointv{1}{i}\) and
    \(A\lointv{j}{|A|} \subseteq S\lointv{i}{|S|}\),
    thus \(A' \subseteq S'\).
    Now for the sake of contradiction, suppose \(A \subseteq S'\).
    It is easy to see that this implies there exists an embedding \(\Psi\) of
    \(A\) into \(S\) which sends \(j\) a value other than \(i\).
    However \(\lembd{A}{S}(i) \leq \Phi(i) \leq \rembd{A}{S}(i)\)
    thus this is not possible and \(A \not \subseteq S'\) as desired.

    We now consider the forward direction.
    As \(i\) is essential for the pair \((A,j)\) in \(S\), we have \(A' \subseteq S'\) thus \(\lembd{A}{S}(j-1) < i < \rembd{A}{S}(j+1)\).
    Note that if there exist an index \(j' \in \ointv{\lembd{A}{S}(j-1)}{ \rembd{A}{S}(j+1)}\) such that \(A[i] = S[j']\), then 
    \(A \subseteq S'\) thus this cannot happen.
    Consequently \(i\) must be the smallest index such that
    \(A\intv{1}{j} \subseteq S\intv{1}{j'}\), thus \(\lembd{A}{S}(j)=i\),
    and must \(i\) also be the largest index such \(A\intv{j}{|A|}\subseteq S\intv{i}{|S|}\), therefore \(\rembd{A}{S}(j)=i\).
\end{proof}

\section{An \(O(n)\) Time Algorithm for Computing an MCS of Two Strings}
Note that we could compute a minimal common supersequence
of $A$ and $B$ by starting out with a maximal common subsequence using Sakai's algorithm \cite{Sakai18}. But that will result in an $O(n\log n)$ time algorithm. We show a linear time algorithm, which is based on reducing a common
supersequence into a minimal one.
The supersequence $A\cdot B$ is a straightforward starting supersequence to obtain a minimal common supersequence. (An example is as follows: $A=abab, B=acbcb$ so initially
$S=A\cdot B = abab\cdot acbcb$. Clearly $S$ can be reduced to $S'=abacbcb$ which is minimal.) This idea can also be generalized to $k$ input strings, with more involved details.

The key idea underlying our algorithms is to sweep through an arbitrary common supersequence $S$ of $A$ and $B$, and, at each position, determine whether the character at that position can be deleted. An index is deemed removable if it is not essential, in the current supersequence, for preserving at least one of the given subsequences.

In order to compute the MCS efficiently, we utilize \cref{lem:ess-check}. Consequently, we construct the image of the right embedding of a sequence \(A\) and \(B\) into its supersequence \(S\). 
We use the following simple procedure. We sweep the sequence \(A\) from right to left and greedily match its indices to indices in \(S\), also from right to left.
We start with the index \(|A|\) in \(A\) then select the largest index \(k_{|A|}\) in \(S\) such that \(S[k_{|A|}] = A[|A|]\). We then match index \(|A|-1\) of \(A\) to the largest index \(k_{|A|-1}\) in \(S\) that occurs before \(k_{|A|}\) such that \(S[k_{|A|-1}] = A[|A|-1]\). 
This process is repeated iteratively until all indices of \(A\) are matched to indices in \(S\).
This right embedding is well defined whenever \(A \subseteq S\). In \cref{alg:2-reduce}, the procedure \(\textsc{BuildRightEmbedding}(S; A)\) returns the image \(r_A\) of this
embedding in ascending order. To be more precise, \( r_A\) stores the sorted indices of $S$ corresponding to the right embedding of $A$ into $S$. For the previous example that
$S=abab\cdot acbcb$ and $A=abab$, we have \(r_A=[3,4,5,9]\).

We now provide details of algorithm \textsc{ReduceSupersequence}.

For the following proofs, we fix a supersequence \(S\)
of strings \(A\) and \(B\).
Further, let \(S_i\) be the subsequence of \(S\) containing indices
where \(\text{hasX}\) is False at the start of the \(i\)-th iteration of the loop starting at 
line~\ref{line:2-main-loop}. Note after the \(i\)-th iteration \(pos = i\). 
Further, let \(d_i\) be the embedding from \(S_i\) into \(S\) where
\(d_i(j)\) gets mapped to the \(j\)-th index
in \(S\) where \(\text{hasX}[j] = \tt{False}\).

\begin{lemma}\label{lem:2-reduce-cor}
    Let \(X \in \{A, B\}\).
    At the start of every iteration of
    the loop on line~\ref{line:2-main-loop} where \(pos=i\),
    we have \(X \subseteq S_i\).
    Further \(l_X = (d_{i} \circ \lembd{X}{S_i})(j_X - 1)\)
    and \(r_X[j_X] = (d_{i} \circ \rembd{X}{S_i})(j_X)\). 
\end{lemma}

\begin{proof}
        Note \(j_X = 1\) and \(l_X = 0\) after line~\ref{line:j-init} and line~\ref{line:l-init},
    respectively.
    Consequently, after entering the loop on line~\ref{line:2-main-loop} with \(pos=0\),
    we have \(l_X = (d_0 \circ \rembd{X}{S_0})(0) \)
    and \(r_X[j_X] = (d_0 \circ \rembd{X}{S_0})(0)\) as \(S_0 = S\).
    
\begin{algorithm}[H]
\caption{ReduceSupersequence}
    \label{alg:2-reduce}
\begin{algorithmic}[1] 
\Function{ReduceSupersequence}{$S, A, B$}
    \State $\text{hasX} \gets [\tt{False}]^{|S|}$. \Comment{array for tracking deletions}
  \State $r_A \gets \Call{BuildRightEmbedding}{S;A} \,\cdot\, [|S|+1]$. \label{line:buildr1}
  \State $r_B \gets \Call{BuildRightEmbedding}{S;B} \,\cdot\, [|S|+1]$. \label{line:buildl1}
  \State $j_B,\, j_A \gets 1$. \Comment{pointer to position \(r_A\) and \(r_B\) } \label{line:j-init}
  \State $l_B,\, l_A \gets 0$. \Comment{index of last index in left embedding} \label{line:l-init}.
  \State $pos \gets 0$. \Comment{position in sweep}
  \While{$pos < |S|$} \label{line:2-main-loop}
  \If{$pos = r_A[j_A]$} \label{line:check-a-stt}\Comment{Update \(j_A, l_A\) if needed}
      \State $l_A \gets l_A + 1$.
      \While{$\text{hasX}[l_A] \text{ or } A[j_A] \ne S[l_A]$}\label{line:a-inc-loop}
        \State $l_A \gets l_A + 1$.
      \EndWhile
      \State $j_A \gets j_A + 1$. \label{line:check-a-end}
    \EndIf
    \If{$pos = r_B[j_B]$} \label{line:check-b-stt}\Comment{Update \(j_B,l_B\) if needed}
      \State $l_B \gets l_B + 1$.
      \While{$\text{hasX}[l_B] \text{ or } B[j_B] \ne S[l_B]$}\label{line:b-inc-loop}
        \State $l_B \gets l_B + 1$.
      \EndWhile
      \State $j_B \gets j_B + 1$. \label{line:check-b-end}
    \EndIf
    \If{$l_A \ne pos \text{ and } l_B \ne pos$}
    \State $\text{hasX}[pos] \gets \tt{True}$. \Comment{remove index if not essential} \label{line:2-del-iness}
    \EndIf
    \State $pos \gets pos + 1$.
  \EndWhile
  \State \Return $\text{concatenate } S[i] \text{ for } i
        \text{ where not } \text{hasX}[i]$. \Comment{generating output}
\EndFunction
\end{algorithmic}
\end{algorithm}

        Suppose the claim is true for some \(i \geq 1\).
    When \(r_X[j_X] \neq i\) the claim holds trivially.
    If \(r_X[j_X] = i\), then either the condition in line~\ref{line:check-a-stt} or
    line~\ref{line:check-b-stt} is true.
    In either case, we continue incrementing
    \(l_X\) until we find the first value where \(\text{hasX}[l_X] = \tt{False}\) and 
    \(S[l_X]=X[j_X]\).
    After satisfying this condition, we increment~\(j_X\) and \(l_X = (d_i \circ \lembd{X}{S_i})(j_X)\).
    After which it follows that \(r_X[j_X] > i\) and
        \(r_X[j_X] = (d_{i} \circ \rembd{X}{S_i})(j_X)\).
    The latter claim follows when \(j_X = |X| + 1\) as any embedding between \(X\) and \(S\) maps \(|X|+1\) to \(|S|+1\).
    Recall that \(\text{hasX}[l_X] = \tt{False}\) and 
    \(S[l_X]=X[j_X]\) and note that if \(X \nsubseteq S_{i+1}\) then \(d_i\) maps
    some value \(t\) to \(i\) which is essential for \(X\) in \(S_{i}\). This cannot be the case as either \(S[l_X]=X[j_X]\) where \(l_X < pos\) or \(j_X = pos \)
    thus \(\text{hasX}[pos] = \tt{False}\). 
    
    The claim is true for \(i+1\) as every iteration increments \(pos\).
    %
\end{proof}

\begin{lemma}\label{lem:2-reduce-li}
     After the \(i\)-th iteration, if \(\text{hasX}[i]=\tt{False}\)
     then, \(t_i\) is essential for either \(A\)
     or \(B\) in \(S_i\) where 
     \(t_i = |\{j \in \intv{1}{i}: \text{hasX}[j]=\tt{False}\}|\).
     Further \(t_i\) is essential for either \(A\) or \(B\) in \(S_l\) for all
     \(l > i\).
\end{lemma}
\begin{proof}
    At the beginning of the iteration, let \(j_A = q\).
    If line~\ref{line:2-del-iness} is not executed, then either \(l_A = pos\) or \(l_B = pos\).
    Without loss of generality, assume that \(l_A = pos\).
    By \cref{lem:2-reduce-cor}, we have \((d_i \circ \lembd{A}{S_i})(q-1) < pos\).
    Consequently, line~\ref{line:check-a-stt} must have been executed in order for \(l_A\) to be incremented.
    It follows that \(r_A[q] = pos = (d_i \circ \lembd{A}{S_i})(q).\)
    Moreover, we have \(\lembd{A}{S_i}(q) = \rembd{A}{S_i}(q) = pos\).
    
    At the end of the iteration, we have \(pos = i\) and \(d_i(i) = t_i\).
    Since \(d_i\) is injective, it follows that \(\lembd{A}{S_i}(q) = t_i = \rembd{A}{S_i}(q)\).
    Therefore, \(t_i\) is essential for the pair \((q, A)\) in \(S_i\).

    Note that, for \(l \geq i \), \(S_i \supseteq S_l\). Moreover, any deletion occurs after the index \(pos\) for all \(l \geq i\). Consequently, \(\rembd{A}{S_l}(q) \leq \rembd{A}{S_i}(q)\) and \(\lembd{A}{S_l}(q) = \lembd{A}{S_i}(q)\); therefore, \(\lembd{A}{S_l}(q)=\rembd{A}{S_l}(q)\).
\end{proof}

\begin{theorem}\label{thm:2-reduce}
    Given a common supersequence \(S\) of strings \(A\) and \(B\).
    A minimal common supersequence \(S' \subseteq S\) can
    be computed in \(O(n)\) time.
\end{theorem}
\begin{proof}
    Note by \cref{lem:2-reduce-li}, every index where
    \(\text{hasX} = \tt{False}\) is essential for either \(A\) or \(B\).
    Thus by \cref{lem:ess-iff}, the algorithm returns an MCS.

    Initialization of the variables requires no more than \(O(n)\) time.
    Every iteration of the while loops starting at line~\ref{line:a-inc-loop}
    or line~\ref{line:b-inc-loop}
    increments \(j_A \in \intv{1}{|A|}\) or \(j_B \in \intv{1}{|B|}\) by one per iteration, thus it takes no more than \(O(|A| + |B|)\) time all together.
    The remaining code in the loop starting at line~\ref{line:2-main-loop} all require no more
    than \(O(1)\) time per iteration thus contribute an additional \(O(n)\) time.
    Therefore the total running time of~\cref{alg:2-reduce} is \(O(n)\).
\end{proof}

As the concatenation of \(A\) and \(B\) is always a common supersequence of both strings,
we have the following result as a consequence of \cref{thm:2-reduce}.


\begin{theorem}
    Given strings \(A\) and \(B\), a minimal common supersequence of \(A\) and
    \(B\) can be computed in \(O(n)\) time.
\end{theorem}

The running of \cref{alg:2-reduce} on the previous example
is as follows. We have \(S=abab\cdot acbcb\), \(A=abab\) and \(B=acbcb\). Initially $r_A=[3,4,5,9]$, $r_B=[5,6,7,8,9]$, and $l_B=l_A=0$. When $pos$ is incremented to 1, the condition at line 19 is satisfied, hence we have $\text{hasX}[1]={\tt True}$. Similarly we have $\text{hasX}[2]={\tt True}$. This means that the first two letters in $S$ will be deleted to finish the reduction.

The details of the algorithm for computing an MCS of \(k\) strings are presented in \cref{sec:kmcs}.

\section{Enumeration of Minimal Common Supersequences}
Our approach to minimal common supersequence enumeration for two strings
\(A\) and \(B\) is to partition them carefully into an equal number of blocks (some could be empty), and then generate the supersequences by using these blocks. These blocks are produced by matching ranges of indices 
in one string to ranges of indices in the other such that one of the
ranges describes a subsequence of the other. An example is given as follows.

\begin{equation*}
A=accdabcdcdab\quad\quad
B=bcbabcdcdcdd
\end{equation*}
\begin{equation*}
A=\fbox{accd}\cdot \fbox{~~~~~~ab}\cdot \fbox{cdcdab}\cdot\fbox{***}
\end{equation*}
\begin{equation*}
B=\fbox{****}\cdot \fbox{bcbab}\cdot \fbox{~~~cdcd}\cdot\fbox{cdd}
\end{equation*}
\begin{equation*}
S=\fbox{accd}\cdot \fbox{bcbab}\cdot \fbox{cdcdab}\cdot\fbox{cdd}
\end{equation*}

In this example, we have $A=A_1\cdot A_2\cdot A_3\cdot A_4$ and $B=B_1\cdot B_2\cdot B_3\cdot B_4$ ($B_1$ and $A_4$ are empty). The minimal common supersequence $S$ is obtained by an alignment of four decomposed blocks of $A$ and $B$, then $S$ is obtained by reading the dominating blocks.


Given two strings \(A,B \in \Sigma^{*}\)
we say an interval \(I \in \interval{A}\) 
\emph{fills} \(J \in \interval{B}\) 
in \(B\) using \(A\) if
\(A[I] \subseteq B[J]\) and no interval containing \(I\) corresponds to a larger 
substring of \(A\) that is a subsequence of \(B[J]\).
(Following the above example, $I=[5,6]$ fills $J=[1,5]$ in $B$ using $A$.) Similarly, we say an interval \(J \in \interval{B}\)
is \emph{full}
for \(I \in \interval{A}\) if \(I\) fills \(J\) in \(B\) using \(A\).
We also say \(J \in \interval{B}\) matches \(I \in \interval{A}\) if  \(B[J] = A[I]\) and we call \(J\) a match in \(B\) for \(I\) in \(A\).

We start by stating some simple observations:
\begin{observation}
 If \(J\) is a match in \(B\) for \(I\) in \(A\),
    then \(J\) in \(B\) is full for \(I\) in \(A\). \label{obs:1}
\end{observation}
\begin{observation}
     If \(I\) fills \(J\) in \(B\) using \(A\) and
    \(I'\) fills \(J'\) in \(B\) using \(A\) then, when 
    \(I \cup I' \in \interval{A}\) and \(J \cup J' \in \interval{J}\),
    we have \(I \cup I'\) fills \(J \cup J'\). \label{obs:2}
\end{observation}


\begin{theorem}\label{thm:bicondtional}
    \(S\) is a minimal common supersequence of strings \(A\) and \(B\) if
    and only if the intervals of \(A\), \(B\), and \(S\) can be partitioned into sequences of
    intervals \((I_i)_{i \in [l]}\), \((J_i)_{i \in [l]}\), and \((K_i)_{i \in [l]}\)
    such that:
    \begin{enumerate}
    \item for every \(i \in [l]\), we have that \(K_i\) is full for
        \(I_i\) in \(A\) and \(J_{i}\) in \(B\).
    \item every \(K_i\) in sequence \((K_i)_{i \in [l]}\) matches
    either \(I_{i}\) in \(A\) or \(J_{i}\) in \(B\).
        \end{enumerate}
\end{theorem}

\begin{proof}
Suppose that \(S\) is a minimal common supersequence of strings \(A\) and \(B\). We aim to show that there exist partitions of \(A\), \(B\), and \(S\) into sequences of intervals \((I_i)_{i \in [l]}, (J_i)_{i \in [l]}\), and \((K_i)_{i \in [l]}\) such that conditions (1) and (2) are satisfied.
We construct the intervals iteratively.
Suppose we already found the first \(j\) intervals where
\(j\) is possibly zero.
Consider the first index \(s\) in \(S\) which does not appear in
any interval we already have.
So as \(S \in \mcsup(A,B)\), we know \(s\) is essential
for \(A\) or \(B\) and assume without loss of generality
it is essential for \(A\).
Take \(K_{j+1}\) to be a maximal interval among intervals containing \(s\) and only the indices essential for \(A\) and which is disjoint with the previous interval \(K_j\). 
Now take \(I_{j+1}\) to be a maximal interval containing every index \(p\) in \([0,|A|+1]\) where some index \(s \in K_{j+1}\) is
essential for \((A,p)\)
and is disjoint with the previous interval \(I_j\). 
Define \(J_{j+1}\) to be the maximal interval containing all
values of \(\intv{1}{|B|}\) that are mapped to \(K_{j+1}\)
in the left embedding.
If no indices are mapped to \(K_{j+1}\), take \(J_{j+1}\)
to be the next available interval of the form \((m,m+1)\) for \(m \in \mathbb{Z}\).
In order for this to be a partition, add \(|A|+1\), \(|B+1|\), \(|S|+1\) to their last interval in each string's respective partition.
Note by construction, every \(K_i\) matches some interval
in \(A\) or \(B\) thus condition (2) holds.
Note by the maximality of the chosen intervals in the construction, \(K_i\) is full for each \(I_i\) and \(j_i\) in \(A\) and \(B\).

    Second we show the backward direction. Suppose we have intervals for \(A, B\) and \(S\) such that conditions (1) and (2) hold. 
    
    Condition (1) implies the following statements:
    \[\prod_{i=1}^{l} S[K_i] \supseteq \prod_{i=1}^{l} A[I_i],\text{ and }\prod_{i=1}^{l} S[K_i] \supseteq \prod_{i=1}^{l} B[J_i].\]
    Note by Condition (2) and Observation 5.2,
    any union of consecutive intervals \(\idseq{K}{i}{l}\) is full for the corresponding union in 
     \(\idseq{I}{i}{l}\) and  \(\idseq{J}{i}{l}\).
     Let \(K_i\) be some interval in \(\idseq{K}{i}{l}\),
     and note that the union of all intervals before
     (and respectively after) are full for 
    their corresponding unions in 
     \(\idseq{I}{i}{l}\) and  \(\idseq{J}{i}{l}\).
     Suppose \emph{w.l.o.g.} that \(K_i\) matches the interval
     \(J_i\) and note deleting any index in
     \(K_i\) cannot produce a supersequence for \(B\)
    as some character in \(B[J_{i}]\) must be added
    to the substrings corresponding to the union of intervals after or before \(J_i\).


    \label{proof:bicondtional}
\end{proof}

Given the partitions for \(A\) and \(B\) it is easy to reconstruct \(S\). But such partitions are not unique thus we are not able to enumerate minimal common supersequences yet.
For uniqueness, we require an additional property to produce a unique partition
of each \(\itv{A}\) and \(\itv{B}\) for a given supersequence.
Luckily, we can do this by adding two conditions: (1) that
intervals in \(\idseq{K}{i}{l}\) are left-closed and are still full 
when their left-endpoints are removed for any interval 
with the same index in \(\idseq{I}{i}{l}\) and \(\idseq{J}{i}{l}\);
(2) intervals in \(\idseq{K}{i}{l}\) do not match at least one string
matched by the prior interval for \(i > 1\).

\begin{theorem}\label{thm:unique-bicondition}
    Let \(S\) be a minimal common supersequence of two distinct strings \(A\) and \(B\).
    There exists exactly one partition of \(\itv{A}\), \(\itv{B}\), and
     \(\itv{S}\) into left-closed intervals such that 
     conditions of the prior theorem hold with the added restrictions:
    \begin{enumerate}
    \item Intervals in \(\idseq{K}{i}{l}\) alternate between the matching intervals sharing their indices in
    \(\idseq{I}{i}{l}\) and \(\idseq{J}{i}{l}\).
    
     \item If \(K_i\) does not match \(I_i\) (or \(J_i\)), 
    we have that the left-open interval \(K_i \setminus \{\lendpt{K_i}\}\) is full for \(I_i\) (respectively, \(J_i\)), where \(\min K_i\) is the minimum of the set \(K_i\).

    \end{enumerate}
\end{theorem}
\begin{proof}
    Note that the condition requiring \(K_{i}\) does not match the same string
    as \(K_{i-1}\), for \(i > 1\), is trivial since when both match
    the identical strings we can combine the \(i\) and \(i-1\) intervals into  a larger interval for each partition.

    We now focus on the second condition. Let \(\idseq{K}{i}{l}\),
    \(\idseq{I}{i}{l}\), and \(\idseq{J}{i}{l}\) be any partitions
    satisfying~\cref{thm:bicondtional} and condition (1).
    Note that \(K_1\setminus \{0\}\) contains the
    same indices of \(S\) as \(K_1\) thus condition (2) trivially holds.
    Furthermore, when the longest common prefix of \(A[I_j]\) and \(S[K_j]\)
    is \(\epsilon\) then condition (2) necessarily holds for any \(j \in [l]\).
    Thus take \(j \in \intv{2}{l}\) to be the first index that violates condition (2) and assume \emph{w.l.o.g.} that \(K_j\) in \(S\) does not match \(I_j\) in \(A\).
    Let \(p_K,p_I,p_J\) be the first index in \(K_j\), \(I_j\), and \(J_j\)
    respectively and let \(I_{j-1}'= I_{j-1} \cup [p_I,p_I+1)\), 
    \(I_{j}' = I_{j}\setminus [p_I, p_I+1)\)  and
    define \(J'_{j-1}\), \(J'_j\), \(K_{j-1}'\), \(K_{j}'\) likewise.
    As condition~(2) fails, we have that \(I_j\) does not fill \(K_j \setminus \{p_K\}\). By the definition of \emph{fill}, this implies that either \(A[I_j] \not\subseteq S[K_j \setminus \{p_K\}]\), or there exists an interval containing \(I_j\) that corresponds to a larger substring of \(A\) which is a subsequence of \(S[K_j]\). However, by assumption, the partition satisfies condition~(2) of \cref{thm:bicondtional}. Therefore, it must be the case that \(A[I_j] \not\subseteq S[K_j \setminus \{p_K\}]\). Consequently, \(A[I_j']\) fills \(S[K_j']\) in \(S\) using \(A\).
    Furthermore, intervals \([p_I,p_I+1)\) and \([p_J,p_J+1)\) must match \([p_K,p_K+1)\) and consequently
    must fill \([p_K,p_K+1)\). As the matching condition holds trivially,
     we may replace \(I_{j-1}, I_{j}, J_{j}, J_{j-1}, K_{j-1}\), and \(K_{j}\)
    with \(I_{j-1}', I_{j}', J_{j}', J_{j-1}', K_{j-1}'\), and \(K_{j}'\)
    in \(\idseq{I}{i}{l}\), \(\idseq{J}{i}{l}\), and  \(\idseq{K}{i}{l}\), respectively to produce new partitions where the common prefix of
    \(j\)th interval of \(\idseq{I}{i}{l}\) and \(\idseq{K}{i}{l}\) is one letter shorter. As we can iterate on this process until
    every interval in  \(\idseq{K}{i}{l}\) not matching an interval with the same index 
    in \(\idseq{I}{i}{l}\) or \(\idseq{J}{i}{l}\) 
    and violating condition (1) have
    differing first characters, we can transform
    the partitions to satisfy condition (1).
    Furthermore, if two partitions have this property but are not the same,
    then the first index in \(S\) which falls in different intervals between the partitions must be in the intervals of \(A\) and \(B\) that match an interval in \(S\) thus it is essential for \(A\) and \(B\) ---
    contradicting condition (1) --- thus the partition is unique. 
\end{proof}

We are now ready to move on to the enumeration algorithm; however, before we do, we prove a
claim that will let us assume \(A\) and \(B\) have no common prefix.
This removes an edge case for the graph structure we produce in the next
section.

\begin{lemma}
    Let \(A = PX\) and \(B = PY\) for some \(P, X, Y \in \Sigma^{*}\).
    Then we have
        \(\mcsup(A, B) = \{ PZ: Z \in \mcsup(X, Y) \}.\)
    \label{lemma:prefixlemma}
\end{lemma}
\begin{proof}
Note it is sufficient to show the claim for \(|P| = 1\).
Let \(U  = \mcsup(X,Y)\) and \(U_P = \{ PZ: Z \in \mcsup(X, Y) \}\).
Note as index \(1\) is essential for \(X\) or \(Y\) in every \(u \in U\),
it is immediate that \(U_P \subseteq \mcsup(A,B)\). Note by~\cref{thm:bicondtional}, we have every \(S \in \mcsup(A,B)\) is of the
form \(S = PW\) for some \(W \in \Sigma^{*}\). If \(S \not \in U_P\),
then there exists \(W' \in U\) where \(W' \subset W\), however then
\(PW' \in \mcsup(A,B)\) where \(PW' \subset S\), which is a contradiction.
    \label{proof:prefixlemma}
\end{proof}

\subsection{Graph for Enumeration}
We now assume that \(A\) and \(B\) have no common prefix other than \(\epsilon\). Using the fact that the characters of \(S\) must match
\(A\) or \(B\), we have the following observation:
\begin{observation}
    \label{obs:3} The first condition of~\cref{thm:bicondtional} 
    can be replaced with
    \(J_i\) fills \(I_i\) when \(K_i\) matches \(I_i\) and
    \(I_i\) fills \(J_i\) when \(K_i\) matches \(J_i\) for
    every \(i \in [l]\). 
\end{observation}

Using this observation, paired with the fact that~\cref{thm:unique-bicondition}
requires that we alternate between matching \(A\) and \(B\) when their
longest common prefix is \(\epsilon\), we define the following
edge-labeled bipartite graph:

\begin{definition}\label{def:mcsup-graph} 
    Given two strings, \(A\) and \(B\),
    define a directed bipartite graph \(\supgraph{A}{B}\)
    with its vertices \(V\) partitioned into \(V_A\) and \(V_B\).
    Let \(X \in \{A,B\}\),
    we define \(V_X\) consisting of tuples of the form \(\ftup{Y}{y}{X}{x}\) where \(x \in \intv{0}{|X|+1}\) and \(y \in \intv{0}{|Y|+1}\) with \(Y \in \{A,B\} \setminus \{X\}\).
    Define \(V_{Y}\) similarly.
    The vertex \(\ftup{X}{x}{Y}{y}\in V_Y\) has an edge to
    \(\ftup{Y}{y'}{X}{x'} \in V_X\) if they satisfy both of the following conditions:
    \begin{description}
        \item[\emph{Closed-fill}:] 
            \((x,x')\) fills \([y, y'+1) \cap \itv{Y}\) in \(Y\) using \(X\),
        \item[\emph{Open-fill}:]
            \((x,x')\) fills \((y, y'+1) \cap \itv{Y}\) in \(Y\) using \(X\).
    \end{description}
    \noindent
    We further define the label of this edge to be the substring in \(Y\)
    corresponding to the interval \([y, y'+1) \cap \itv{Y}\), which we refer to as the interval of this edge.
\end{definition}

Given a path in \(\supgraph{A}{B}\), the label of the path is the concatenation of the substrings labeling its edges
    following the order of the edges appearing
    in the path.
    We call the vertices \((A,0,B,0)\) and \((B,0,A,0)\) the start nodes in
    \(\supgraph{A}{B}\) and
    \((A,|A|+1,B,|B|+1)\), \((B,|B|+1,A,|A|+1)\) the end nodes.
    We define the \(st\)-paths of \(\supgraph{A}{B}\) to be the paths from a
    start node to an end node.
    Finally we define the graph \(\stsupgraph{A}{B}\) as the subgraph of
    \(\supgraph{A}{B}\) induced by the vertices of the \(st\)-paths.
    We now show that the labels of \(st\)-paths in \(\supgraph{A}{B}\)
    correspond exactly to the minimal common supersequences of \(A\) and \(B\).

\begin{figure}
    \centering
    \includegraphics[width=\textwidth]{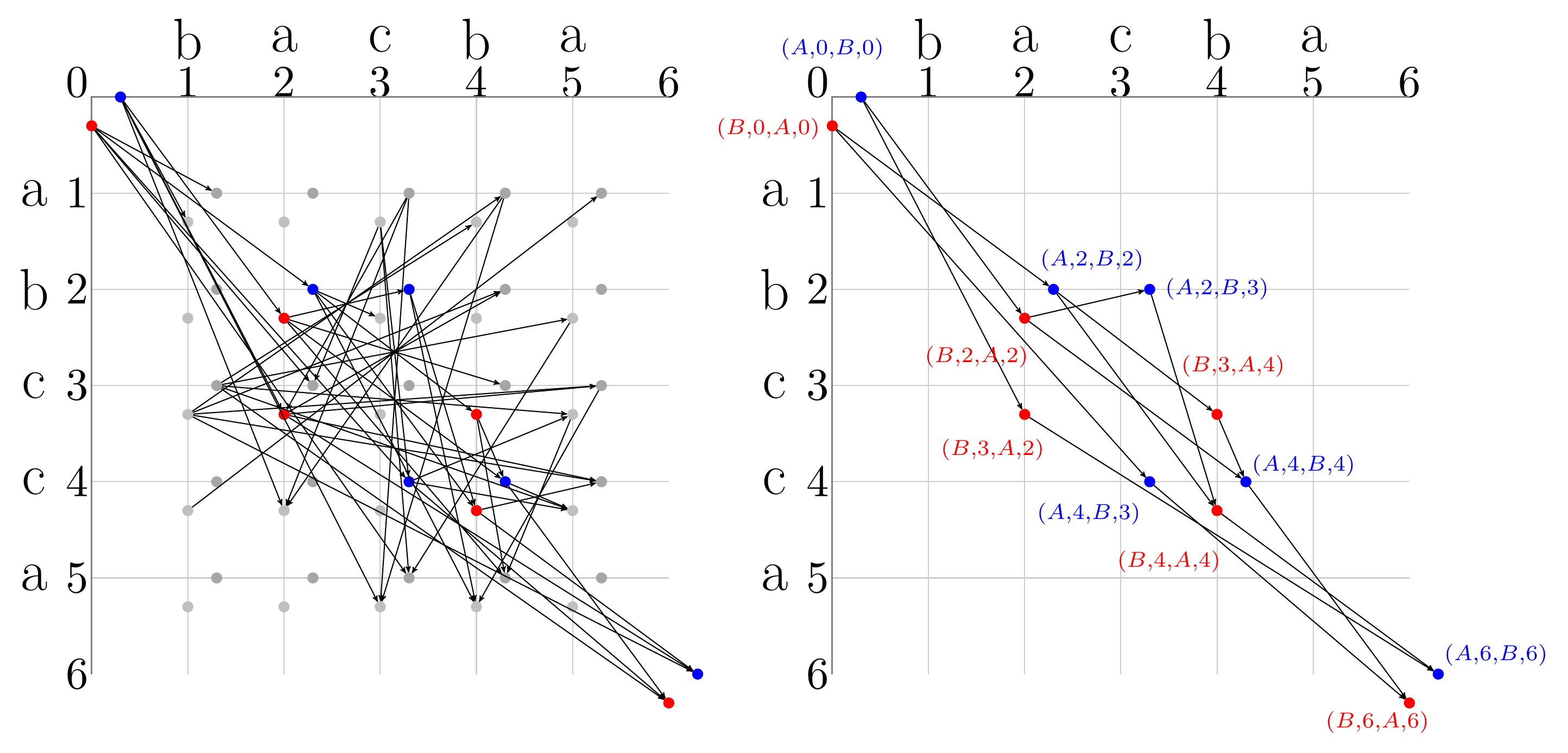}
    \caption{Depicted on the left is \(\supgraph{A}{B}\) containing all nodes
    with the vertices of \(\stsupgraph{A}{B}\) colored in terms of their respective partitions. Red vertices belong partition \(V_A\) and blue vertices belong to partition \(V_B\). On the right we have \(\stsupgraph{A}{B}\)
    with the vertices labeled. \(A=bacba, B=abcca\).}
    \label{fig:enum-graph}
\end{figure}
    
\begin{theorem}
    There exists a bijection between the \(st\)-paths of \(\supgraph{A}{B}\) and the set \(\mcsup({A},{B})\).
\end{theorem}
\begin{proof}
    Let \(S\) be the label of some \(st\)-path in \(\supgraph{A}{B}\).
    Using Observation 5.3, we have \(S \in \mcsup(A,B)\).
    Let \(L_i\) be the index of the first character in the label of the \(i\)-th edge in \(S\) and \(l_i\) be the number of indices in the interval for the \(i\)-th edge. 
    Let \(K_i = [L_i, L_i+l_i+1) \cap \itv{S}\) and note that it corresponds to the substring added by the \(i\)-th edge.
    Further the closed-fill and open-fill condition show that \(\idseq{K}{i}{l}\) must satisfy condition (1) of \cref{thm:unique-bicondition} and, as \(\supgraph{A}{B}\) is bipartite, condition (2) also holds.
    Thus the labels of the \(st\)-paths define
    an injection into \(\mcsup(A,B)\).
    We see it is surjective since any \(W \in \mcsup(A,B)\) can be partitioned into intervals that match the strings which correspond to the labels of \(\supgraph{A}{B}\) by~\Cref{thm:unique-bicondition,thm:bicondtional}. 
\end{proof}

\subsection{Speeding Up the Construction of \(\stsupgraph{A}{B}\)}

Our primary strategy for speeding up the construction of 
\(\stsupgraph{A}{B}\) is to find intervals of \(A\) and 
\(B\) that meet the open-full and closed-full conditions. We show that these conditions can easily be computed via dynamic programming.
In this end, we define 
\[l_{X,Y}(x, y, y') = \max \{x' \in [|X|]: X[(x,x')] \subseteq Y[[y,y']]\}.\]
We note the following recurrence relation.
\begin{equation}\label{ref:recur-rel}
    l_{X,Y}(x,y,y'+1) =
    \begin{cases}
        l_{X,Y}(x, y, y') + 1 \quad &\text{if }X[l_{X,Y}(x,y,y')+1] = Y[y'+1]\\
        l_{X,Y}(x, y, y')  &\text{otherwise.}
    \end{cases}
\end{equation}
When \(y = y'-1\), we have \(l_{X,Y}(x,y,y') = x + 1 \text{ for any } x \in [|X|].\)
More importantly, the following observations allow us to find intervals of \(X\) which fill a given interval of \(Y\).

\begin{lemma}\label{lem:fil-cond}
    Let \(X,Y \in \{A,B\}\) such that \(X \neq Y\).
    Let \((x,x') \in \interval{X}\) and let \([y, y'+1) \in \interval{Y}\)
    that only use indices of \(X\) and \(Y\) respectively.
    We have \((x, x')\) fills \([y, y'+1)\) in \(Y\) using \(X\)
    if and only if
    \begin{equation}
                x' = l_{X,Y}(x,y,y') > l_{X,Y}(x-1,y,y').  
    \end{equation}
\end{lemma}
\begin{proof}
    Suppose \((x,x')\) fills \([y,y'+1)\) in \(Y\) using \(X\),
    then we have \(X[(x,x')]\subseteq  Y[[y,y'+1)] \),
    \(X[[x,x')] \nsubseteq Y[[y,y'+1)]\), and 
    \(X[(x,x']] \nsubseteq Y[[y,y'+1)]\).
    Thus we have \(x' = l_{X,Y}(x,y,y')\) and
    \(x' > l_{X,Y}(x-1,y,y')\).

    Now suppose \(x' = l_{X,Y}(x,y,y') > l_{X,Y}(x-1,y,y')\).
    As \(x' = l_{X,Y}(x,y,y')\), we know \( X[(x,x')] \subseteq Y[[y,y'+1)]\)
    and \(X[(x,x']] \nsubseteq Y[[y,y'+1)]\).
    Since \( x' > l_{X,Y}(x-1,y,y')  \),  we see \(X[[x,x')] \nsubseteq Y[[y,y'+1)]\).
    Thus \((x,x')\) fills \([y,y'+1)\) in \(Y\) using \(X\).
\end{proof}

We further have the following result to determine the edge of \(\supgraph{A}{B}\)
that only use indices of \(G\) by plugging in the appropriate values to~\cref{lem:fil-cond}.

\begin{corollary}\label{lem:dyn-edge-condition}
    Let \(X,Y \in \{A,B\}\) such that \(X \neq Y\).
    For \(x,x' \in [|X|]\) and \(y,y' \in [Y]\),
    we have \((X,x,Y,y)\) has an edge to \((Y,y',X,x')\) in
    \(\supgraph{A}{B}\) if and only if
    \begin{align}
         x'   &= l_{X,Y}(x,y,y'), \\
         x' & > l_{X,Y}(x-1,y,y'), \\
         x' &= l_{X,Y}(x,y+1,y'), \\
         x' & > l_{X,Y}(x-1,y+1,y').  
    \end{align}
\end{corollary}
It is easy to adapt the prior statement to handle the case when the values
are not indices of \(A\) and \(B\) by setting them to 
the nearest index in their respective string.
Furthermore, we note that for a given vertex in \(\supgraph{A}{B}\)
at most \(\max\{|A|,|B|\}\) values for \(y'\) can satisfy~\cref{lem:dyn-edge-condition}.    
Let \(n = |A| + |B|\). As a consequence we have the following corollary.

\begin{corollary}
    The number of edges in \(\supgraph{A}{B}\) satisfies \(|E(\supgraph{A}{B})| =\ O(n^3)\) and no node has more than
    \(O(n)\) (out-going) edges.
\end{corollary}

We further extend this to the following lemma.

\begin{lemma}\label{lem:graph-const}
    The nodes of \(\stsupgraph{A}{B}\) can be found in \(O(n^3)\)
    time using only \(O(n^2)\) space.
\end{lemma}
\begin{proof}
    Note that \(|V(\supgraph{A}{B})| = O(n^2)\).
    In order to compute which vertices are in \(st\)-paths, we
    preform a depth-first search to determine the vertices connected to the start and end nodes while computing edges of \(\stsupgraph{A}{B}\) as needed per vertex. 
    To avoid producing the edge lists for the vertices, we save only the last \(y'\) value used to find an edge for each vertex in the graph.
    Thus we only use \(O(n^2)\) space during the computation.
    We are able to compute all edges of \(\supgraph{A}{B}\) in \(O(n^3)\) time
    by spending at most \(O(n)\) time computing the (out-going) edges of each vertex
    using~\cref{ref:recur-rel}.
    Furthermore, as we spend \(O(n^3)\)
    time per depth-first search, the algorithm run in \(O(n^3)\) time.
\end{proof}

We now state the main result.

\begin{theorem}
    There exists an algorithm to enumerate minimal common supersequences of strings \(A\) and \(B\) 
    using \(O(n^3)\) time preprocessing, \(O(n^2)\) space, and 
    \(O(n)\) delay.
\end{theorem}
\begin{proof}
    The enumeration algorithm begins by computing the vertices in \(\stsupgraph{A}{B}\) using the method described in~\cref{lem:graph-const}.
    For each vertex in \(\stsupgraph{A}{B}\), we then spend an additional \(O(n)\) time to find the largest \(y'\) for which the vertex has an edge.
    All this can be done in \(O(n^3)\) time, resulting in \(O(n^3)\) preprocessing
    time.
    We then find all paths from a start node to an end node using depth-first search by computing the edges as needed as in~\cref{lem:graph-const}, with the
    modification that we stop once the largest \(y'\) value for the node is reached.

    Every time we compute a value using~\cref{ref:recur-rel} and
    do not find a new edge, we are able to add one character to the next
    sequence.
    This follows as this index must be used in the interval of the next edge in the path.
    As any minimal common supersequence is shorter than
    \(|A| + |B|\), we cannot spend more than \(O(n)\) time computing the edges
    between two outputs. Further, by a similar reasoning, we 
    cannot spend more than \(O(n)\) time backtracking during the depth-first
    search. Thus the algorithm has \(O(n)\) delay.
    The \(O(n^2)\) space requirement comes from the space required to
    store the nodes, the \(y'\) value corresponding to their last edge,
    and the space required for preprocessing.
\end{proof}

\section{An \(O(kn\log n )\) Time Algorithm for Computing an MCS of \(k\) Strings} \label{sec:kmcs}
We now move to reducing a common supersequence of \(k\) strings to a
minimal common supersequence.
It is straightforward to modify~\cref{alg:2-reduce} to accommodate \(k\) strings and achieve an algorithm with \(O(k^2n)\) runtime.
As it is meaningful to consider minimal common supersequences among many small strings, it makes sense to consider algorithms which avoid the heavy dependence on \(k\). 

We make two changes to improve the runtime when \(k\) is large.
First, we store our
output in a data structure that allows us efficiently
search the first occurrence of character after a given index.
This allows us avoid sweeping through \(S\) for every string as we do
from line~\ref{line:check-a-stt} to line~\ref{line:check-b-end} in~\cref{alg:2-reduce}.
Second, we store all right embedding indices of the supersequence together with the identifiers of their corresponding strings in a single array, where the pairs are sorted in ascending order by index. This organization allows us to update the indices associated with left embeddings without having to repeatedly scan the separate arrays that store the right embeddings for each string.

In the pseudocode of~\cref{alg:k-reduce}, we denote the procedure that constructs this sorted array for a supersequence \(S\) of strings \(A_1, \ldots, A_k\) by\\
\(\textsc{MergeRightEmbeddings}(S\,;\, A_1, \ldots, A_k)\). It is straightforward to see that this array can be constructed in \(O(N \log N)\) time where \(N\) is the sum of the lengths of the \(k\) input strings. For easier comprehension we give the output of \textsc{MergeRightEmbeddings} for two strings \(A_1=abbc\) and \(A_2=ac\) for the supersequence \(S=abccbacc\) in \cref{fig:ds-fig}. 



\subsubsection{A Data Structure for Searching for Occurrences}
The structure we consider implicitly stores a subsequence \(S\) of \(T\in \Sigma^*\) by keeping track of
three things: (1) the length of \(|S|\), (2) the number of occurrences of
each character in the string, and (3) for each character \(\sigma \in \Sigma\), an array which stores
in ascending order the indices where this character occurs in \(T\).
Let \(n_\sigma\) be the number of occurrences of \(\sigma\) in \(T\).
The structure further pre-allocates an array of size \(n_\sigma\), for each \(\sigma\).
\Cref{fig:ds-fig} depicts an example of the occurrence arrays for
a subsequence \(S\) of string \(T\).
To append the character \(\sigma\) to \(S\), we update the structure by
writing the value \(|S|+1\) into the next available cell in the array corresponding to \(\sigma\). This update can be done in \(O(1)\) time
as it only requires the position of the next available cell and length \(|S|\).
Given an index \(i\) in \(S\) and a character \(c\), we perform a binary search on the occurrence array associated with \(c\) to determine the smallest index \(j > i\) such that \(S[j] = c\).
If no such index exists, we return \(|S|+1\).

Lastly, note that \(S\) can be reconstructed from the data structure with the aid of the sequence \(T\). We iterate over the characters of \(T\), and for each character \(c\), we lookup its corresponding occurrence array in the data structure. If the next index of \(S\) appears in that occurrence array, we output \(c\); otherwise, we proceed to the next character of \(T\) and repeat the process. This reconstruction procedure runs in \(O(|T|)\) time in the worst case.


\begin{figure}
    \includegraphics[width=1.0\textwidth]{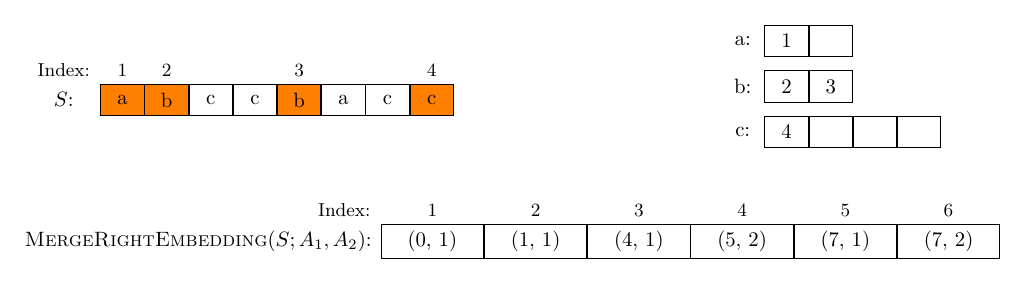}
    \caption{
        On the left hand side, we show characters used in a right embedding 
        of \(A_1 =abbc\) into \(S=abccbacc\) in orange cells.
        On the right, we depict the output data structure for \(S\) and
        \(A_1\).
        On the bottom is the result of \(\textsc{MergeRightEmbedding}(S;A_1,A_2)\) where \(A_2=ac\).
    }    \label{fig:ds-fig}
\end{figure}

We summarize the results for this in the following lemma: 
\begin{lemma}
    \label{lem:ds-lem}
    Given a sequence \(T\) containing \(n_\sigma \in \Z_{>0}\) occurrences 
    of \(\sigma \in \Sigma\), there exists a  
    data structure that can be created in \(O(|T|)\) time
    which stores a string \(S \subseteq T\) and provides the following operations:
    \begin{itemize}
        \item a method \(\textsc{Insert}(\sigma)\) that appends \(\sigma\) to
            \(S\) in \(O(1)\) time as long as \(\sigma\) does
            not appear more than \(n_\sigma\) times in \(S\cdot\sigma.\)
    \item a method \(\textsc{FindNext}(\sigma, i)\) that returns the next occurrence of \(\sigma\) after index \(i\)
            or \(|S|+1\) if no such occurrence exits.
        \item a method \(\textsc{BuildStr()}\) that
            produces a copy of \(S\) in \(O(n)\) time.
    \end{itemize}
\end{lemma}

In our psuedocode, we write \(\textsc{OccArrString}\) to
denote the function constructing the data structure in~\cref{lem:ds-lem}
from an input sequence.
Further we write \(|D|\) to denote the size of string stored in the structure \(D\).


%

\begin{algorithm}
\caption{Algorithm for reducing a supersequence of \(k\)}
\label{alg:k-reduce}
\begin{algorithmic}[1]
\Function{ReduceToMinimal}{$S\,;\,A_1 \ldots A_k$}
  \State $outp \gets \textsc{OccArrString}\big(\textsc{GetOccuranceData}(\mathit{S})\big)$ \label{line:k-start-init}
  \State $l\_all[i] \gets 0$ for all $i \in \{1,\dots,k\}$
  \State $r\_all \gets \textsc{MergeRightEmbeddings}(S\,;\, A_1, \ldots A_k)$ \label{line:k-end-init}
  \State $pos \gets 1$
  \State $j\_all \gets 1$
  \While{$pos \leq |S|$} \label{line:k-main-loop}
    \State $(r\_idx, str\_id) \gets r\_all[j\_all]$
    \State $c \gets S[pos]$
    \If{$pos = r\_idx$} \label{line:all-stt}
      \State $new\_top \gets outp.\textsc{FindNext}(c, l\_all[str\_id])$\label{line:k-find-next}.
      \If{$new\_top = |outp|+1$}
        \State $outp.\textsc{Insert}(c)$ \label{line:k-insert}
      \EndIf
      \State $l\_all[str\_id] \gets new\_top$ \label{line:l-inc}
      \State $j\_all \gets j\_all + 1$ \label{line:j-inc}
    \Else
    \State $pos \gets pos + 1$ \label{line:incrementing}
    \EndIf
  \EndWhile
  \State \Return $outp.\textsc{BuildStr()}$
\EndFunction
\end{algorithmic}
\end{algorithm}

\subsection{Reducing a Common Supersequence for \(k\) Strings}

In the following proofs, let \(P_l\) denote the string stored
in \(outp\) and \(S_l = P_l \cdot S\intv{pos}{|S|}\) at the start of
the \(l\)-th iteration of the loop on line~\ref{line:k-main-loop}. 
The largest value of the left embedding
of each string is stored in an array of size \(k\) named \(l\_all\).
The variable \(r\_all\) keeps all values of
the right embeddings in a single array
as describe earlier.
The value of \(j\_all\) maintains the location in the array \(r\_all\) as 
we sweep through \(S\) such that \(r\_all[j\_all]\) is the next possible index which
we can add to the output.

The following lemma shows that \(l\_all\) and \(j\_all\) are correctly maintained during our algorithm's execution.
\begin{lemma}
    For any \(i \in [k]\), let \(j_i\) be the smallest value such that \((\rembd{A_i}{S}(j_i),i)\) is in \(r\_all\intv{j\_all}{}\). 
    At the start of line~\ref{line:k-main-loop} of the \(l\)-th iteration, we have \(A_i \subseteq S_l\) and 
    the following are true:
    \begin{enumerate}
    \item \(\rembd{A_i}{S_l}(j_i) = \rembd{A_i}{S}(j_i)-(pos-|P_l|-1), \)
    \item \(l\_all[i]=\lembd{A_i}{S_l}(j_i-1).\)

    \end{enumerate}
    \label{lem:k-str-loop-prop}
\end{lemma}
\begin{proof}
Suppose that at the beginning of each iteration we have \( i = \textit{str\_id} \), and let \( r\_idx \) denote the right-embedding position of index \( j_i \) in string \( A_i \). Recall that \( j\_all\) stores all right-embedding indices of the supersequence together with the identifiers of the corresponding strings in a single array, sorted in increasing order of the indices.

In the first iteration, we process the first element \( (r\_idx, \textit{str\_id}) \) of \(j\_all \). At this point, we have,
 \[
 l\_all[i] = 0 \quad \text{for all } i \in [k].
 \]
 Since no elements of \( r\_all \) have been processed yet, for every string \( A_i \) we have \( j_i = 1 \) and hence \( j_i - 1 = 0 \). Moreover,
 \[
 S_1 = P_1 \cdot S[1:|S|] = S,
 \]
since \( P_1 \) is the empty string. Therefore, condition~(1) holds immediately, as
 \[
 \rembd{A_i}{S_1}(j_i)
 = \rembd{A_i}{S}(j_i) - (pos - |P_1| - 1)
 = \rembd{A_i}{S}(j_i) - (1 - 0 - 1)
 = \rembd{A_i}{S}(j_i).
 \]
 Condition~(2) also holds trivially in the first iteration, since \( j_i - 1 = 0 \) and \(\lembd{A_i}{S_1}(0) = 0 = l\_all[i]\).

 At the end of iteration \( l \), we either increment the variable \( pos \) or update \( l\_all[i] \) for the selected string \( A_i \) by assigning it the left-embedding value of index \( j_i \). This update occurs only when the condition \( pos = r\_idx \) is satisfied; otherwise, we increment \( pos \) until the condition holds. Each time we skip a position in \( S \), we implicitly delete the corresponding character from the output.

 Once \( pos = r\_idx \), we attempt to match index \( j_i \) of string \( A_i \) to the next available position in \( P_l \). To find this position, we run a
 \(\textsc{findnext}(c, l\_all[i])\)
 query. If there exists an unmatched occurrence of the character \( c \) in \( P_l \), we update \(l\_all[i] \) to that position. Otherwise, we match the character at position \( pos \) and update \( l\_all[i] \) accordingly. In all cases, \( l\_all[i] \) stores the left-embedding position of index \( j_i \).

 In the next iteration, \( l+1 \), we process the subsequent element of \( j_{\text{all}} \), which again corresponds to some string \( A_i \) and a new index \( j_i \). Condition~(1) continues to hold: since index \( j_i - 1 \) has already been matched to a position in \( P_{l+1} \), we can compute the right-embedding position of \( j_i \) in \( S_{l+1} \) by taking \( \rembd{A_i}{S}(j_i) \) and subtracting \(pos - |P_{l+1}| - 1,\)
 which accounts for the characters implicitly deleted from \( S \).

 Finally, condition~(2) also holds, because index \( j_i - 1 \) was matched in an earlier iteration that processed string \( A_i \). Consequently, the current value of \(l\_all[i] \) equals \(\lembd{A_i}{S_{l+1}}(j_i - 1),\) as required.
 \end{proof}
\begin{lemma}
    During the \(l\)-th iteration, line~\ref{line:k-insert} is executed
    if and only if \(t_l = |P_l|+1\)  is essential for some \(A_j\) for \(S_{l'}\) for every
    \(l'\geq l\).
    ~\label{lem:k-str-loop-inv}
\end{lemma}
\begin{proof}
 We start with the only if direction.
 Note that line 13 can only be executed after lines 8 to 13.
 Let \(i = str\_id\) and suppose has line 13 is executed.
 Note \(\rembd{A_{i}}{S}(j_i)=r\_idx\) where \(j_i\) is first index
 in \(A\) where \((\rembd{A_i}{S}(j_i),i)\in r\_all\intv{j\_all}{}\).
 After line 11 is executed, we set \(new\_top\) 
 to the next occurrence of \(c=A_i[j_i]\) after \(l\_all[i]\) thus
 \(new\_top = \lembd{A_{j}}{S_l}(j_i) \) by~\cref{lem:k-str-loop-prop}.
 Furthermore \(\lembd{A_{j}}{S_l}(j_i)=|P_l|+1\) as \(|outp|+1=|P_l|+1\) and line 12.
 \Cref{lem:k-str-loop-prop} is applicable again to see
 \(\rembd{A_i}{S_l}(j_i) = \rembd{A_i}{S}(j_i)-(pos-|P_l|-1)
 =|P_l|+1.\) Thus we see \(|P_l|+1\) is essential for \(A_j\) in \(S_l\).
 Finally, as \(\rembd{A_i}{S_{l'}}(j_i)\) may only decrease with respect to \(l'\) as \(S_{l'} \subseteq S_l\), we see
 \(\lembd{A_i}{S_{l'}}(j_i)=\rembd{A_i}{S_{l'}}(j_i)\)
 for all \(l' \geq l\).

 Now let \(|P_l|+1\) be essential for some pair \((j_i, A_i)\) in \(S_l\) thus \(\rembd{A_i}{S_l}(j_i) = |P_l|+1\).
 \Cref{lem:k-str-loop-prop} now implies \(j_i\) is in \(r\_all\intv{j\_all}{}\) as
 \(l\_all[i] \leq |P_l|\) at all
 times during execution as a consequence. Furthermore, it must be the first value for \(A_j\) in
 \(r\_all\intv{j\_all}{}\) as the right embedding is injective and, as another consequence of~\cref{lem:k-str-loop-prop}, we \(\rembd{A_j}{S}(j_i) = pos\).
\end{proof}
    


After the last iteration of~\cref{alg:k-reduce}, we are left
with \(P_l = S_l\).
Using this we now derive our main result for \(k\)-strings.

\begin{theorem}
    Given a common supersequence \(S\) of \(k\) strings \(A_1,\ldots, A_k \), a minimal common supersequence \(S' \subseteq S\) can
    be computed in \(O(N\log N)\) time where \(N\) is the total length of the input strings.
\end{theorem}

\begin{proof}
    By~\cref{lem:k-str-loop-inv} and~\cref{lem:k-str-loop-prop},
    every index in the output is essential for some \(A_i\), \(i \in [k]\), and every \(A_i\) is a subsequence of the output thus the
    algorithm computes an MCS correctly.
    
    We now examine the runtime. 
    Initialization of data structures from line~\ref{line:k-start-init} 
    to line~\ref{line:k-end-init} is linear hence the time
    complexity of these operations is dominated by the \(O(N \log N)\) time of $MergeRight$.
    Then, note that the iterations of~\cref{alg:k-reduce} can be partitioned
    into iterations which update \(pos\) and those that do not.
    Iterations which increments \(pos\) each takes \(O(1)\) time thus contribute
    \(O(N)\) time to execution.
    Iterations that do not increment \(pos\), increment \(j\_all[i]\)
    for some \(i\) and spend \(O(\log(N))\) time computing
    \textsc{FindNext} at line~\ref{line:k-find-next} by~\cref{lem:ds-lem}.
    For any \(i\), \(j\_all[i] \leq |A_i|\), we thus make
    \(\sum_{i\in[k]} |A_i| = N\) calls to \textsc{FindNext} 
    for a total of \(O(N\log N)\) time.
    Consequently the loop at line~\ref{line:k-main-loop} also
    requires \(O(N\log N)\) time to compute, thus the algorithm
    takes \(O(N \log N )\) time overall.
\end{proof}
As \(N \leq kn\), we arrive at the following time bound in terms of \(n\) and \(k\).

\begin{theorem}
    Given \(k\) strings \(A_1, \ldots, A_k\), 
    a minimal common supersequence of them
    can be found in \(O(kn(\log k+\log n))\) times.
\end{theorem}

\section{Concluding Remarks}

We present a linear time algorithm to compute a minimal common supersequence between two strings. For \(k\) input
strings, each with length at most $n$, a similar method
gives an $O(kn(\log k+\log n))$ time algorithm. For the problem of enumerating minimal common supersequences on two strings each with length at most $n$, we built an $O(n^2)$-space data structure in $O(n^3)$ time such that each minimal common supersequence can be enumerated with an $O(n)$ time delay. An interesting question is if some constrained version can also be computed efficiently; for example, what if the computed minimal common supersequence must not contain a string $P$ as a subsequence.


\end{document}